\documentclass[a4paper,UKenglish]{lipics-v2016}
\usepackage{microtype}%if unwanted, comment out or use option "draft"
\usepackage[inference]{semantic}
\usepackage{stmaryrd}

\theoremstyle{plain}
\newtheorem{proposition}[theorem]{Proposition}

\title{A Uniform Approach to Random Process Model}
\titlerunning{A Uniform Approach to Random Process Model} %optional, in case that the title is too long; the running title should fit into the top page column

\author[1]{Yuxi Fu}
\affil[1]{BASICS, Shanghai Jiao Tong University, China\\
  \texttt{fu-yx@cs.sjtu.edu.cn}}
\authorrunning{Yuxi Fu} %mandatory. First: Use abbreviated first/middle names. Second (only in severe cases): Use first author plus 'et. al.'

\Copyright{Yuxi Fu}%mandatory, please use full first names. LIPIcs license is "CC-BY";  http://creativecommons.org/licenses/by/3.0/

\subjclass{F.1.1 Models of Computation, F.1.2 Modes of Computation}% mandatory: Please choose ACM 1998 classifications from http://www.acm.org/about/class/ccs98-html . E.g., cite as "F.1.1 Models of Computation".
\keywords{probabilistic process, branching bisimulation, divergence}% mandatory: Please provide 1-5 keywords
\begin{document}

\maketitle

\begin{abstract}
There is a lot of research on probabilistic transition systems.
There are not many studies in probabilistic process models.
The lack of investigation into the interactive aspect of probabilistic processes is mainly due to the difficulty caused by the discrepancy between probabilistic actions and nondeterministic behaviours.
The paper proposes a uniform approach to probabilistic process models and a bisimulation congruence for probabilistic concurrency.
\end{abstract}

\section{Introduction}\label{Introduction}

Randomization plays an indispensable role in computer science.
The celebrated result, the PCP Theorem~\cite{AroraLundMotwaniSudanSzegedy1992}, reveals the power of ``interaction+randomness+error'' in problem solving.
Given an NP complete problem, one may design an interactive proof system consisting of a verifier and a prover~\cite{GoldwasserMicaliRackoff1985,Babai1985}.
Upon receiving a problem instant the verifier accepts or rejects the input with high confidence in polynomial time by using logarithmic random bits and asking a constant number of questions to the prover.
The scenario can be generalized to a multi-prover situation with an increased power on the verifier side~\cite{Ben-OrGoldwasserKilianWigderson1988,BabaiFortnowLund1991,FortnowRompelSipser1994}.
This fundamental result is significant to modern computing systems, which are open, distributed, interactive, and have both nondeterministic behaviours and randomized choices.
To formalize models in which results like the PCP Theorem apply, one may introduce randomization to interaction models (process models).
There are two kinds of randomness in randomized process models.
A process may send a random value to another; and it may randomly choose whom it will send  a value to.
We call the former {\em content randomness} and the latter {\em channel randomness}.
Content randomness is basically a computational issue~\cite{Vadhan2012,FuYuxi2017}, whereas channel randomness is to do with interaction.

What kind of channel randomness are there?
In literature one finds basically two answers to the question~\cite{LarsenSkou1989,HanssonJonsson1989,WangLarsen1992,Segala1995,vanGlabbeekSmolkaSteffen1995}.
{\em Generative models} feature probabilistic choice for external actions.
The standard syntax for a probabilistic choice term is of the form
\begin{equation}\label{probability-operator}
\bigoplus_{i\in I}p_i\ell_i.T_i,
\end{equation}
where $p_i\in(0,1)$ and $\Sigma_{i\in I}p_i=1$.
The infix notation $p_1\ell_1.T_1\oplus\ldots\oplus p_k\ell_K.T_k$ is often used.
The semantics is defined by $\bigoplus_{i\in I}p_i\ell.T_i\stackrel{\ell_i}{\longrightarrow}_{p_i}T_i$, meaning that $\bigoplus_{i\in I}p_i\ell.T_i$ may evolve into $T_i$ with probability $p_i$ by performing the action $\ell_i$.
The generative model is problematic in the presence of the {\em interleaving} composition operator and the localization operator.
Let $A$ be $\frac{1}{2}a\oplus\frac{1}{2}b$ and $C$ be $\frac{2}{3}\overline{b}\oplus\frac{1}{3}\overline{c}$.
What is then the behaviour of $A\,|\,C$?
And how about $(b)(A\,|\,C)$ and $(a)(c)(A\,|\,C)$?
What is the probability of $A$ interacting with $C$ at channel $b$ in $(a)(c)(A\,|\,C)$?
In $(c)C$ interaction at channel $c$ is disabled.
How does that reconcile with the prescription that $C$ interacts at channel $c$ with probability $1/3$?
It does not sound right to say that $(c)C$ performs the $\overline{b}$ action with probability one.
A reasonable semantics is that $(c)C$ may do the $\overline{b}$ action with probability $2/3$ and becomes dead with probability $1/3$.
If this is indeed the interpretation, $C$ should really be $\frac{2}{3}\overline{b}\oplus\frac{1}{3}\tau.\overline{c}$.
Symmetrically one may argue that $\frac{2}{3}\overline{b}\oplus\frac{1}{3}\tau.\overline{c}$ should really be $\frac{2}{3}\tau.\overline{b}\oplus\frac{1}{3}\tau.\overline{c}$.
All problems with the probabilistic choice~(\ref{probability-operator}) is gone if it is replaced by the {\em random choice} term
\begin{equation}\label{random-operator}
\bigoplus_{i\in I}p_i\tau.T_i,
\end{equation}
where the size of the index set $I$ is at least $2$ and $\sum_{i\in I}p_i=1$.
Thus $0<p_i<1$ for all $i\in I$.
Early generative models are {\em fully probabilistic}~\cite{BaierHermanns1997}.
Nondeterminism was considered later~\cite{Segala1995}.

In reactive models, introduced by Larsen and Skou~\cite{LarsenSkou1989} and popularized by the work of van Glabbeek, Smolka and Steffen~\cite{vanGlabbeekSmolkaSteffen1995}, nondeterministic choice and probabilistic choice come in alternation.
Using a suggestive notation one may write for example
\begin{equation}\label{2019-01-05}
a.\left(\frac{1}{2}A_1\uplus\frac{1}{2}A_2\right)+b.\left(\frac{1}{3}B_1\uplus\frac{2}{3}B_2\right).
\end{equation}
This is a process that may perform an $a$ action and turns into $A_1$ with probability $1/2$ and $A_2$  with probability $1/2$.
It may also do an interaction at channel $b$ and becomes $B_1$ with probability $1/3$ and $B_2$  with probability $2/3$.
It is not helpful to think of $\frac{1}{2}A_1\uplus\frac{1}{2}A_2$ simply as a distribution over $\{A_1,A_2\}$.
The distribution can only be achieved by carrying out a certain amount of computation, say invoking a random number generator.
The details of the computation can be abstracted away, but it should definitely be formalized as an internal action.
The best way to understand the process in~(\ref{2019-01-05}) is to see it as a simplification of
\begin{equation}\label{2019-01-15}
a.\left(\frac{1}{2}\tau.A_1\oplus\frac{1}{2}\tau.A_2\right)+b.\left(\frac{1}{3}\tau.B_1\oplus\frac{2}{3}\tau.B_2\right).
\end{equation}
The process in~(\ref{2019-01-15}) may do an external nondeterministic choice, and then an internal random choice.
This is why reactive models are also called (strict) alternating models.
However once we have separated the two kinds of choice, there is no point in insisting on the alternation.
What it means is that we might as well give up on generative probabilistic choice and reactive probabilistic choice altogether in favour of~(\ref{random-operator}) and nondeterministic choice.
A systematic exposure of the research progress on reactive models is given in Deng's excellent book~\cite{Deng2015}.

The central issue in defining a probabilistic process model is the treatment of nondeterminism in the presence of probabilistic choice.
The philosophy we shall be following in this paper is that nondeterminism is an attribute of interaction while randomness is a computational feature.
Nondeterminism is a system feature, which cannot be implemented.
Randomness is a process property, which can be implemented with a negligible error.
We advocate a model independent methodology that turns an interaction model into a randomized interaction model by adjoining~(\ref{random-operator}).
The semantics of the random operator is defined by
\begin{equation}\label{oplus-rule}
\inference{}{\bigoplus_{i\in I}p_i\tau.T_i \stackrel{p_i\tau}{\longrightarrow}T_i}.
\end{equation}
We emphasize that the label $p_i\tau$ should be understood as the same thing as $\tau$.
The additional information attached by $p_i$ is to help reasoning with the bisimulation semantics.
Talking about bisimulation equivalence it is often useful to think of the transitions defined by~(\ref{oplus-rule}) as a single silent transition.
We introduce the {\em collective silent transition}
\begin{equation}\label{2019-01-16}
\bigoplus_{i\in I}p_i\tau.T_i\stackrel{\coprod_{i\in I}p_i\tau}{\longrightarrow}\coprod_{i\in I}T_i.
\end{equation}
The collective silent transition is closed under composition, localization and recursion.

Strong bisimulations for probabilistic labeled transition systems, pLTS for short, are well understood~\cite{LarsenSkou1989,HanssonJonsson1989,Segala1995,vanGlabbeekSmolkaSteffen1995,Deng2015}.
Weak bisimulations have been studied for reactive models~\cite{SegalaLynch1994,Deng2015} and alternation models~\cite{Philippou1LeeSokolsky2000}.
In the presence of probabilistic choice a silent transition sequence appears as a tree of silent transitions.
Schedulers are introduced to resolve the nondeterminism when constructing such a tree.
Branching bisimulations have also been studied for reactive models~\cite{SegalaLynch1994}.

Our current understanding of weak/branching bisimulations for probabilistic process models is not very satisfactory in several accounts, which can be summarized as follows.
\begin{itemize}
\item
Majority of the works are about pLTS.
In a pLTS process combinators disappear.
As far as we know none of the weak/branching bisimulations studied in literature is closed under all the three indispensable process combinators, the composition, localization and recursion operators.
In fact some of them is closed in none of the three operators.
This is not surprising because a pLTS without referring to any model defines a semantics for automata~\cite{Segala2006}, not a semantics for processes.
There are suggestions to look at synchronous probabilistic process models~\cite{vanGlabbeekSmolkaSteffen1995,BaierHermanns1997}.
A basic problem in the synchronous scenario is if internal actions are synchronized.
A yes answer seems to contradict to the very idea of observational theory.
But if the silent transitions are not synchronized, the composition operator is unlikely associative.
\item
A consequence of the failure to account for the composition and localization, most results, even definitions, apply to only finite state probabilistic processes~\cite{Philippou1LeeSokolsky2000,AndovaWillemse2006,Deng2015}.
The coincidence between the weak bisimularity and the branching bisimilarity for example is only proved for the finite state fully probabilistic processes~\cite{BaierHermanns1997}.
In fact in literature probabilistic processes are often defined as finite labeled graphs~\cite{Philippou1LeeSokolsky2000} or labeled concurrent Markov chains~\cite{Vardi1985}.
These restricted models preempt any study on process combinators.
\item
The issue of divergence has not been properly dealt with.
This is definitely an omission, especially so in the presence of random silent actions.
\end{itemize}

The main task of the paper is to justify the model independent methodology proposed in the above.
We shall convince the reader not only that randomization of process calculi ought to be model independent, but also that the bisimulation theory of the randomized version of any process model $\mathbb{M}$ can be obtained from the bisimulation theory of $\mathbb{M}$ in a uniform manner.
Section~\ref{RCCS} defines a randomized process model.
For simplicity the model is taken to be a sub-model of Milner's CCS.
Section~\ref{Epsilon-Tree} introduces $\epsilon$-tree and showcases its role in transferring the bisimulation theory of a model $\mathbb{M}$ to the bisimulation theory of randomized $\mathbb{M}$.
Section~\ref{Equality4RCCS} proves the congruence property of the bisimulation equivalence.
Section~\ref{Comment} makes some final comment.

\section{Random Process Model}\label{RCCS}

Let $Chan$ be the set of channels, ranged over by lowercase letters.
Let $\overline{Chan}=\{\overline{a}\mid a\in Chan\}$.
The set $Chan\cup\overline{Chan}$ will be ranged over by small Greek letters.
We let $\overline{\alpha}=a$ if $\alpha=\overline{a}$.
The set of actions is $Act=Chan\cup\overline{Chan}\cup\{\tau\}$.
We write $\ell$ and its decorated versions for elements of $Act$.
The grammar of CCS~\cite{Milner1989} is as follows:
\begin{equation}\label{2019-01-25}
S,T \;:=\; X \mid \sum_{i\in I}\alpha_i.T_i \mid S\,|\,T \mid (a)T \mid \mu X.T,
\end{equation}
where the indexing set $I$ is finite.
We write ${\bf 0}$ for the {\em nondeterministic term} $\sum_{i\in\emptyset}\alpha_i.T_i$ in which $\emptyset$ is the empty set.
A trailing ${\bf 0}$ is often omitted.
We also use the infix notation of $\sum$, writing for example $\alpha_1.T_1+\alpha_2.T_2+\alpha_3.T_3$.
A process variable $X$ that appears in $\sum_{i\in I}\alpha_i.T_i$ is guarded.
We shall assume that in the {\em fixpoint term} $\mu X.T$ the {\em bounded} variable $X$ is guarded in $T$.
A term is a {\em process} if it contains no free variables.
We write $A,B,C,D,E,F,G,H$ for processes.
Let $\mathcal{T}_{\mathrm{CCS}}$ be the set of all CCS terms and $\mathcal{P}_{\mathrm{CCS}}$ be the set of all CCS processes.
A {\em finite state} term/process is a term/process that contains neither the composition operator nor the localization operator.
We can define $\tau$-prefix in the standard manner.
For example $a.A+\tau.B$ can be defined by $(c)(\overline{c}\,|\,(a.A+c.B))$ for some fresh $c$.
From now on we shall use this derived notation without further comment.
The transition semantics of CCS is generated by the following rules, where $\lambda\in Act$.
\vspace*{-2mm}
\[\begin{array}{ccc}
\inference{}{\sum_{i\in I}\alpha_i.T_i \stackrel{\alpha_i}{\longrightarrow} T_i}\ \ \ \ &
\inference{S\stackrel{\overline{\alpha}}{\longrightarrow} S'\ \ \ \ T\stackrel{\alpha}{\longrightarrow} T'}{S\,|\,T\stackrel{\tau}{\longrightarrow} S'\,|\,T'}\ \ \ \ &
\inference{T\stackrel{\lambda}{\longrightarrow} T'}{S\,|\,T\stackrel{\lambda}{\longrightarrow} S\,|\,T'}
\end{array}\]
\[\begin{array}{ccc}
\inference{S\stackrel{\lambda}{\longrightarrow}S'}{S\,|\,T\stackrel{\lambda}{\longrightarrow} S'\,|\,T}\ \ \ \ &
\inference{T\stackrel{\lambda}{\longrightarrow}T'}{(a)T\stackrel{\lambda}{\longrightarrow} (a)T'}\ a\notin\lambda\ \ \ \ &
\inference{T\{\mu X.T/X\}\stackrel{\lambda}{\longrightarrow} T'}{\mu X.T\stackrel{\lambda}{\longrightarrow} T'}
\end{array}\]

\vspace*{1mm}
For an equivalence $\mathcal{E}$ on $\mathcal{P}_{\mathrm{CCS}}$ we write $A\,\mathcal{E}B$ for $(A,B)\in\mathcal{E}$.
The advantage of the infix notation is that we may write for example $A\mathcal{E}B\mathcal{E}C$ and $A\stackrel{\ell}{\longrightarrow}B\mathcal{E}C$.
The notation $\mathcal{P}_{\mathrm{CCS}}/\mathcal{E}$ stands for the set of equivalence classes defined by $\mathcal{E}$.
The equivalence class containing $A$ is denoted by $[A]_{\mathcal{E}}$, or $[A]$ when the equivalence is clear from context.
We write $A\stackrel{\tau}{\longrightarrow}_{\mathcal{E}}A'$ if $A\stackrel{\tau}{\longrightarrow}A'\mathcal{E}A$, and $\Longrightarrow_{\mathcal{E}}$ for the reflexive and transitive closure of $\stackrel{\tau}{\longrightarrow}_{\mathcal{E}}$.
For $\mathcal{C}\in\mathcal{P}_{\mathrm{CCS}}/\mathcal{E}$ we write $A\stackrel{\ell}{\longrightarrow}\mathcal{C}$ for the fact that $A\stackrel{\ell}{\longrightarrow}A'\in\mathcal{C}$ for some $A'$.
A process $A$ is $\mathcal{E}$-{\em divergent} if there is an infinite silent sequence $A\stackrel{\tau}{\longrightarrow}_{[A]_{\mathcal{E}}}\ldots\stackrel{\tau}{\longrightarrow}_{[A]_{\mathcal{E}}}\ldots$.

The {\em Randomized CCS}, RCCS for short, is defined on top of CCS.
The RCCS terms are obtained by extending the definition in~(\ref{2019-01-25}) with the randomized choice term defined in~(\ref{random-operator}).
A variable that appears in $\bigoplus_{i\in I}p_i\tau.T_i$ is also guarded.
The transition semantics of RCCS is defined by the above rules of CCS plus the rule defined in~(\ref{oplus-rule}).
The label $\lambda$ that appears in these rules ranges over $Act\cup\{p\tau\mid 0<p<1\}$.
The set of RCCS terms is denoted by $\mathcal{T}_{\mathrm{RCCS}}$ and that of RCCS processes by $\mathcal{P}_{\mathrm{RCCS}}$.

We shall find it convenient to interpret $T\stackrel{1\tau}{\longrightarrow}T'$ as $T\stackrel{\tau}{\longrightarrow}T'$.
So $\stackrel{p\tau}{\longrightarrow}$ is a random silent transition if $0<p<1$ and an interaction if $p=1$.
The (reflexive and) transitive closure of $\stackrel{\tau}{\longrightarrow}$ is denoted by $\stackrel{\tau}{\Longrightarrow}$ ($\Longrightarrow$).
We shall say that $p_1\ldots p_k$ is the probability of the silent transition sequence $T\stackrel{p_1\tau}{\longrightarrow}\ldots\stackrel{p_k\tau}{\longrightarrow}T'$.

\section{Epsilon Tree}\label{Epsilon-Tree}

Bisimulation equivalence is the standard equality for concurrent objects~\cite{Milner1989,Park1981}.
Bauer and Hermanns' technique~\cite{BaierHermanns1997} applied in the proof that the weak bisimilarity coincides with the branching bisimilarity on the finite-state fully probabilistic processes offers a convincing argument that one should focus on the branching bisimulation equivalence in probabilistic setting.
For any process equality $\asymp$ on $\mathcal{P}_{\mathrm{CCS}}$ one thinks of a silent transition $A\stackrel{\tau}{\longrightarrow}_{\asymp}A'$ as {\em state-preserving}, and a silent transition $A\stackrel{\tau}{\longrightarrow}A'$ such that $A'\not\asymp A$ as {\em state-changing}.
The basic idea of van Glabbeek and Weijland's branching bisimulation~\cite{vanGlabbeekWeijland1989-first-paper-bb,vanGlabbeekWeijland1996} is that a state-changing silent action must be explicitly bisimulated whereas state-preserving silent actions are ignorable.
If $B\asymp A\stackrel{\tau}{\longrightarrow}_{\asymp}A'$ then $B$ does not have to do anything because $B\asymp A'$.
If $B\asymp A\stackrel{\tau}{\longrightarrow}A'\not\asymp A$ then $A\stackrel{\tau}{\longrightarrow}A'$ must be simulated by some $B\stackrel{\tau}{\Longrightarrow}B'$.
Branching bisimulation requires that conversely $B\stackrel{\tau}{\Longrightarrow}B'$ must be simulated by $A\stackrel{\tau}{\longrightarrow}A'$.
It is in this sense that $A\stackrel{\tau}{\longrightarrow}A'$ is {\em bi}simulated by $B\stackrel{\tau}{\Longrightarrow}B'$.
The difference between branching bisimilarity and weak bisimilarity is that the former is a bisimulation equivalence whereas the latter is a simulation equivalence.
A minute's thought would lead us to believe that $B\stackrel{\tau}{\Longrightarrow}B'$ must be of the form $B\Longrightarrow_{\asymp}\stackrel{\tau}{\longrightarrow}B'\asymp A'$.
With these remarks in mind let us formalize the notion of branching bisimulation.
\begin{definition}\label{2018-12-23}
An equivalence $\mathcal{E}$ on $\mathcal{P}_{\mathrm{CCS}}$ is a {\em branching bisimulation} if for all $\ell$ and all $\mathcal{C}\in\mathcal{P}_{\mathrm{CCS}}/\mathcal{E}$ such that $\ell\ne\tau\vee\mathcal{C}\ne[A]$, the following statement is valid for all $A,B\in\mathcal{P}_{\mathrm{CCS}}$.
\begin{itemize}
\item If $B\,\mathcal{E}A\Longrightarrow_{\mathcal{E}}\stackrel{\ell}{\longrightarrow}\mathcal{C}$, then $B\Longrightarrow_{\mathcal{E}} \stackrel{\ell}{\longrightarrow}\mathcal{C}$.
\end{itemize}
\end{definition}
\vspace*{1mm}

\noindent Clearly $B\,\mathcal{E}\,A\stackrel{\tau}{\longrightarrow}_{\mathcal{E}}A'$ implies $B\,\mathcal{E}\,A'$.
It follows from definition that $A\stackrel{\tau}{\longrightarrow}_{\mathcal{E}}A'$ is bisimulated by $B$ vacuously.
That explains the condition $\ell\ne\tau\vee\mathcal{C}\ne[A]$.

The extensional equality for computation never identifies a nonterminating computation to a terminating computation. The best way to formalize this requirement in bisimulation semantics is introduced in~\cite{Priese1978}.
It is the key condition that turns a bisimulation equality for interaction to an equality for both interaction and computation~\cite{FuYuxi2016}.
\begin{definition}\label{2019-01-12}
An equivalence $\mathcal{E}$ on $\mathcal{P}_{\mathrm{CCS}}$ is {\em codivergent} if, for every $\mathcal{C}\in\mathcal{P}_{\mathrm{CCS}}/\mathcal{E}$, either all members of $\mathcal{C}$ are $\mathcal{E}$-divergent, or no member of $\mathcal{C}$ is $\mathcal{E}$-divergent.
\end{definition}
The union of a class of codivergent branching bisimulations on $\mathcal{P}_{\mathrm{CCS}}$ is a codivergent branching bisimulation on $\mathcal{P}_{\mathrm{CCS}}$~\cite{FuYuxi2016}.
So we may let $=_{\mathrm{CCS}}$ be the largest such relation on $\mathcal{P}_{\mathrm{CCS}}$.

Having motivated the bisimulation equality for CCS, we are in a position to randomize it as it were to an equality for RCCS.
In RCCS a silent transition is generally a distribution over a finite set of silent transitions.
A finite sequence of silent transitions in CCS then turns into a silent transition tree in RCCS.
To describe that we introduce an auxiliary definition.
\begin{definition}
Suppose $\mathcal{E}$ is an equivalence on $\mathcal{P}_{\mathrm{RCCS}}$ and $A\in\mathcal{P}_{\mathrm{RCCS}}$.
A {\em silent tree $t$ of $A$} is a labeled tree rendering true the following statements.
\begin{itemize}
\item Every node of $t$ is labeled by an element of $\mathcal{P}_{\mathrm{RCCS}}$.
The root of $t$ is labeled by $A$.
\item The edges are labeled by elements of $(0,1]$.
If an edge from a node labeled $A'$ to a node labeled $A''$ is labeled $p$, then $A'\stackrel{p\tau}{\longrightarrow}A''$.
\end{itemize}
An {\em $\mathcal{E}$-tree $t^A$ of $A$} is a silent tree of $A$ such that all the labels of the nodes of $t^A$ are in $[A]_{\mathcal{E}}$.
\end{definition}
If we confuse a node with its label, we may say for example that $A'\stackrel{q}{\longrightarrow}A''$ is an edge in $t^A$.
Definition~{\ref{2019-02-03} formalizes state-preserving silent transition sequence in the probabilistic setting.
\begin{definition}\label{2019-02-03}
An {\em $\epsilon$-tree $t_{\mathcal{E}}^A$ of $A$ with regard to $\mathcal{E}$} is an $\mathcal{E}$-tree of $A$ rendering true (\ref{xh},\ref{XH}).
\begin{enumerate}
\item \label{xh} If $B\stackrel{q}{\longrightarrow}B'$ for some $q\in(0,1)$, some collective silent transition $B\stackrel{\coprod_{i\in [k]}p_i\tau}{\longrightarrow}\coprod_{i\in[k]}B_i$ exists such that $B\stackrel{p_i}{\longrightarrow}B_i$ for all $i\in[k]$ and $B_1,\ldots,B_k$ are the only children of $B$.
\item \label{XH} If $B\stackrel{1}{\longrightarrow}B'$, then $B\stackrel{\tau}{\longrightarrow}B'$ and $B'$ is the only child of $B$.
\end{enumerate}
\end{definition}
Intuitively an $\epsilon$-tree of $A$ with regard to $\mathcal{E}$ is a random version of $\Longrightarrow_{\mathcal{E}}$.
All nodes of an $\epsilon$-tree with regard to $\mathcal{E}$ are equal from the viewpoint of $\mathcal{E}$.
Condition~\ref{xh} requires that if one of $B_1,\ldots,B_k$ is in the $\epsilon$-tree then all of $B_1,\ldots,B_k$ are in the $\epsilon$-tree, and $B\stackrel{q}{\longrightarrow}B'$ is $B\stackrel{p_i}{\longrightarrow}B_i$ for some $i\in I$.
This is nothing more than the intuition that $B\stackrel{\coprod_{i\in [k]}p_i\tau}{\longrightarrow}\coprod_{i\in[k]}B_i$ is conceptually a single silent transition.
The number of $\epsilon$-trees of $A$ with regard to an equivalence class are in general infinite.
Let's see some examples.

\begin{example}\label{ex-1}
Let $\Omega_a=\mu X.(\tau.a+\tau.X)$.
Let $\mathcal{E}_1$ be any equivalence that distinguishes a divergent process from a non-divergent one.
A finite $\epsilon$-tree of $\Omega_a$ with regard to $\mathcal{E}_1$ corresponds to a finite transition sequence of the form
$\Omega_a\stackrel{\tau}{\longrightarrow}\Omega_a\stackrel{\tau}{\longrightarrow}\ldots\stackrel{\tau}{\longrightarrow}\Omega_a$.
In the non-random case an $\epsilon$-tree with regard to $\mathcal{E}_1$ is just an instance of $\Longrightarrow_{\mathcal{E}_1}$.
There is an infinite $\epsilon$-tree of $\Omega_a$, corresponding to the divergent sequence $\Omega_a\stackrel{\tau}{\longrightarrow}\Omega_a\stackrel{\tau}{\longrightarrow}\ldots$.
\end{example}

\begin{example}\label{ex-2}
Let $\Omega_{\frac{1}{2}}=\mu X.(\frac{1}{2}\tau.X\oplus\frac{1}{2}\tau.X)$.
There are infinitely many $\epsilon$-trees of $a\,|\,\Omega_{\frac{1}{2}}$ with regard to any equivalence.
An $\epsilon$-tree may be a single node tree (the left diagram below), or a three node tree (the middle diagram below), or an infinite tree (the right diagram below).
Unlike Example~\ref{ex-1} the divergence in this case is immune from any intervention.
\begin{center}
\includegraphics[scale=0.5]{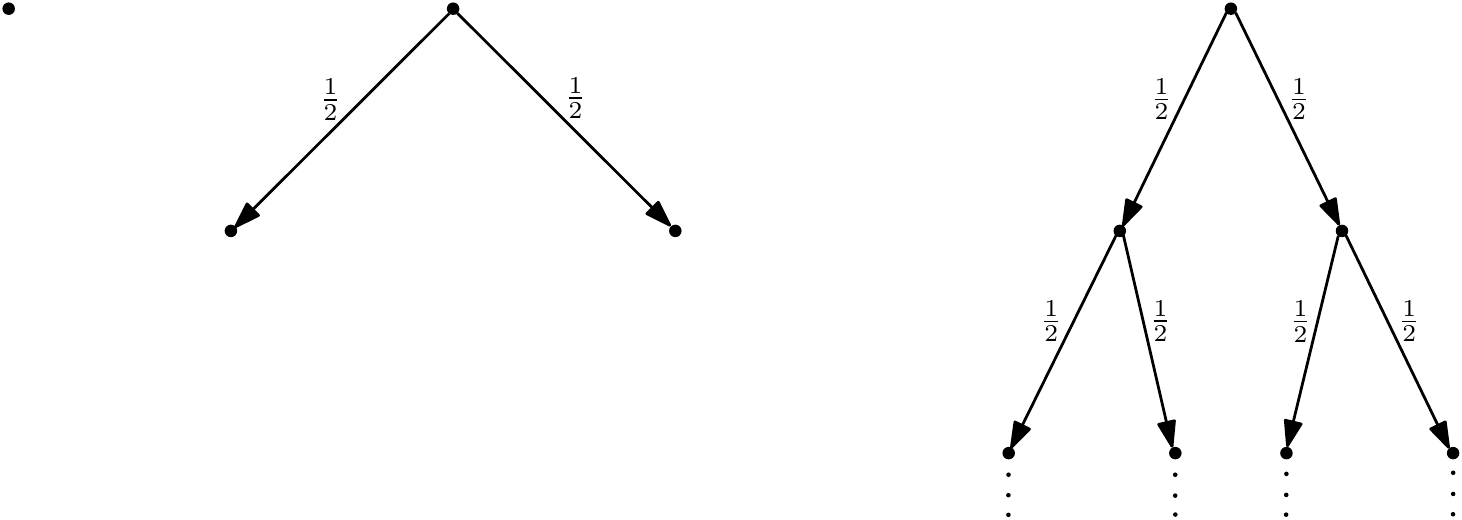}
\end{center}
\end{example}

\begin{example}\label{ex-3}
Let $\Omega_{\frac{1}{2}a}=\mu X.(\frac{1}{2}\tau.a\oplus\frac{1}{2}\tau.X)$.
Let $\mathcal{E}_2$ be an equivalence such that $[\Omega_{\frac{1}{2}a}]_{\mathcal{E}_2}=[a]_{\mathcal{E}_2}$.
A finite $\epsilon$-tree of $\Omega_{\frac{1}{2}a}$ with regard to $\mathcal{E}_2$ is described by the left diagram below, one of its leaves cannot do an immediate $a$ action.
The right diagram describes an infinite $\epsilon$-tree of $\Omega_{\frac{1}{2}a}$ with regard to $\mathcal{E}_2$, all of its leaves can do an immediate $a$ action.
\begin{center}
\includegraphics[scale=0.5]{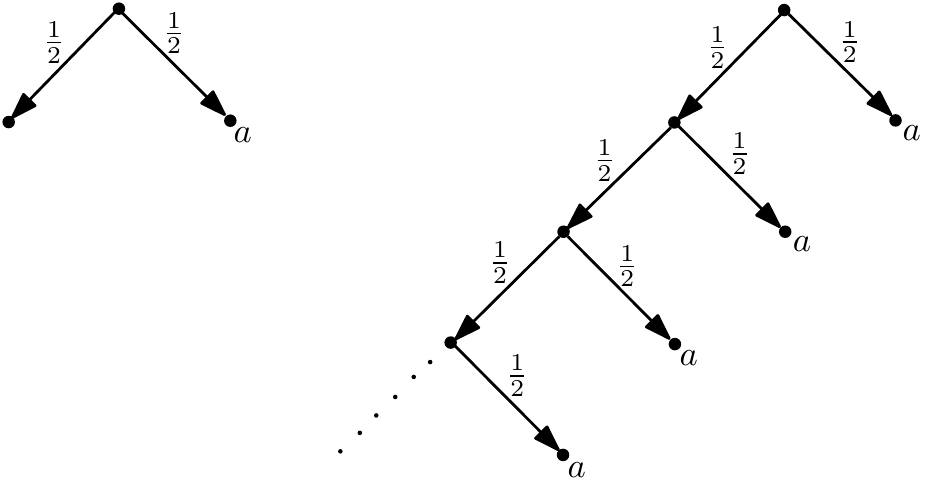}
\end{center}
\end{example}

\begin{example}\label{ex-4}
Let $G=\mu X.(\frac{1}{3}\tau.(a+\tau.X)\oplus\frac{2}{3}\tau.X)$.
Let $\mathcal{E}_3$ be any equivalence such that $[G]_{\mathcal{E}_3}=[a+\tau.G]_{\mathcal{E}_3}$.
Two $\epsilon$-trees of $G$ with regard to $\mathcal{E}_3$ are described by the following infinite diagrams.
Every leaf of the left diagram can do an immediate $a$ action, whereas none of the leaves of the right diagram can do an immediate $a$ action.
\begin{center}
\includegraphics[scale=0.5]{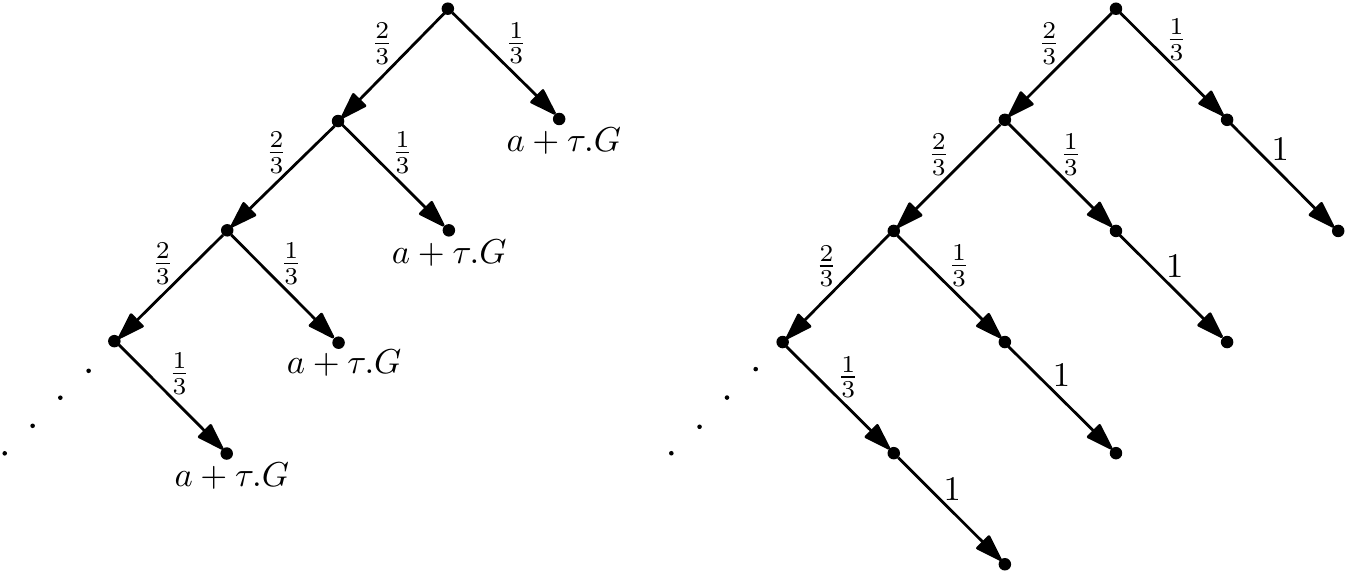}
\end{center}
\end{example}

\begin{example}\label{ex-5}
Let $H=\mu X.(\frac{1}{2}\tau.(a+\tau.X)\oplus\frac{1}{2}\tau.(b+\tau.X))$.
An $\epsilon$-tree of $H$ with regard to an equivalence $\mathcal{E}_4$ rendering true $[H]_{\mathcal{E}_4}=[a+\tau.H]_{\mathcal{E}_4}=[b+\tau.H]_{\mathcal{E}_4}$ is described by the left diagram below.
Every leaf of the $\epsilon$-tree can do an immediate $a$ action.
Another $\epsilon$-tree of $H$ with regard to $\mathcal{E}_4$ is described by the right diagram below, in which every leaf can do an immediate $b$ action.
\begin{center}
\includegraphics[scale=0.5]{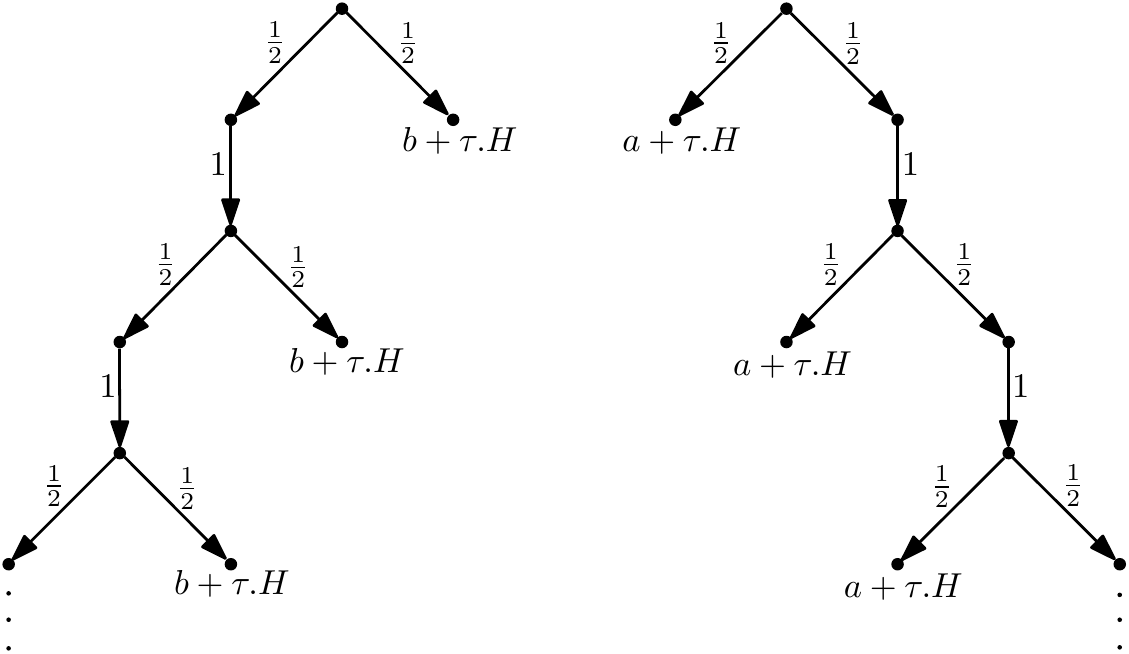}
\end{center}
\end{example}

These examples bring out a few observations.
Firstly $\epsilon$-trees are meant to generalize $\Longrightarrow_{\mathcal{E}}$.
This is clear from Example~\ref{ex-1}.
However $\epsilon$-trees are a little too general.
Two $\epsilon$-trees of a process may differ in that every leaf of one $\epsilon$-tree may do an immediate $a$ action whereas in the other this is not true.

To isolate the $\epsilon$-trees that truly correspond to $\Longrightarrow_{\mathcal{E}}$, we introduce some auxiliary  definitions.
A {\em path} in a silent tree $t$ is either a finite path going from the root to a node or an infinite path starting from the root.
A {\em branch} of $t$ is either a path ending in a leaf or an infinite path.
The length $|\pi|$ of a path $\pi$ is the number of edges in $\pi$ if $\pi$ is finite; it is $\omega$ otherwise.
For $i\le|\pi|$ let $\pi(i)$ be the label of the $i$-th edge.
The probability $\textsf{P}(\pi)$ of a finite path $\pi$ is $\prod\{\pi(i) \mid i\in[|\pi|]\}$.
A path of length zero is a single node, and its probability is $1$.
The probability of an infinite path $A\stackrel{p_1\tau}{\longrightarrow}\stackrel{p_2\tau}{\longrightarrow}\ldots\stackrel{p_k\tau}{\longrightarrow}\ldots$
is the limit of $p_1,p_1p_2,\ldots,\prod_{i\le k}p_i,\ldots$, whose existence is guaranteed because the decreasing sequence is bounded by $0$ from below.
It $t$ is finite, define $\textsf{P}(t)=\sum\{\textsf{P}(\pi)\mid \pi\ \mathrm{is}\ \mathrm{a}\ \mathrm{branch}\ \mathrm{of}\ t\}$.
If $t$ is infinite, we need to define the probability in terms of approximation.
Let $t{\upharpoonright}_k$ be the subtree of $t$ defined by the nodes of height no more than $k$.
Inductively
\begin{itemize}
\item $t{\upharpoonright}_0$ is induced by the root of $t$; and
\item $t{\upharpoonright}_{k+1}$ is induced by the nodes of $t{\upharpoonright}_k$ and all the children of these nodes.
\end{itemize}
It should be clear that $\textsf{P}(t{\upharpoonright}_{k+1})\le\textsf{P}(t{\upharpoonright}_k)$.
The probability $\textsf{P}(t)$ of the tree $t$ is defined by the limit $\lim_{k\rightarrow\infty}\textsf{P}(t{\upharpoonright}_k)$.
\begin{lemma}
$\textsf{P}(t)=1$ for every $\epsilon$-tree $t$.
\end{lemma}
\begin{proof}
$\textsf{P}(t{\upharpoonright}_{k})=1$ for all $k\ge0$.
\end{proof}
The probability of the finite branches of $t$ is defined by $\textsf{P}^f(t)=\lim_{k\rightarrow\infty}\textsf{P}^k(t)$, where
\begin{equation}\label{2019-01-13}
\textsf{P}^k(t)=\sum\left\{\textsf{P}(\pi) \mid \pi\ \mathrm{is}\ \mathrm{a}\ \mathrm{finite}\ \mathrm{branch}\ \mathrm{in}\ t\ \mathrm{such}\ \mathrm{that}\ |\pi|\le k\right\}.
\end{equation}

We are now in a position to generalize a branching bisimulation for CCS processes to a branching bisimulation for RCCS processes.
First of all we generalize state-preserving silent transition sequences of finite length.
Intuitively such a sequence turns into an $\epsilon$-tree that probabilistically contains no infinite branches.
\begin{definition}
An $\epsilon$-tree $t^A_{\mathcal{E}}$ is {\em regular} if $\textsf{P}^f(t^A_{\mathcal{E}})=1$.
\end{definition}
In the same line of thinking an $\epsilon$-tree is divergent if it has no finite branches.
\begin{definition}
An $\epsilon$-tree $t^A_{\mathcal{E}}$ is {\em divergent} if $\textsf{P}^f(t^A_{\mathcal{E}})=0$.
\end{definition}
The next definition is the probabilistic counterpart of Definition~\ref{2019-01-12}.
\begin{definition}
An equivalence $\mathcal{E}$ on $\mathcal{P}_{\mathrm{RCCS}}$ is {\em codivergent} if the following is valid:
\begin{itemize}
\item For every $\mathcal{C}\in\mathcal{P}/\mathcal{E}$, either all members of $\mathcal{C}$ have divergent $\epsilon$-trees with regard to $\mathcal{E}$, or no member of $\mathcal{C}$ has any divergent $\epsilon$-tree with regard to $\mathcal{E}$.
\end{itemize}
\end{definition}
To discuss the branching bisimulation for random processes, we need to talk about a transition from a process $A$ to an equivalence class $\mathcal{B}\in\mathcal{P}/\mathcal{E}$.
This makes sense because the processes in $\mathcal{B}$ are supposed to be all equal.
We would like to formalize the idea that after a finite number of state-preserving silent transitions an $\ell$-action is performed and the end processes are in $\mathcal{B}$.
Suppose $\ell\ne\tau\vee \mathcal{B}\ne[A]$.
An {\em $\ell$-transition from $A$ to $\mathcal{B}$ with regard to $\mathcal{E}$} consists of a regular $\epsilon$-tree $t_{\mathcal{E}}^A$ of $A$ with regard to $\mathcal{E}$ and a transition $L\stackrel{\ell}{\longrightarrow}L'\in\mathcal{B}$ for every leaf $L$ of $t_{\mathcal{E}}^A$.
We will write $A\rightsquigarrow_{\mathcal{E}}\stackrel{\ell}{\longrightarrow}\mathcal{B}$ if there is an $\ell$-transition from $A$ to $\mathcal{B}$ with regard to $\mathcal{E}$.
By definition $A\rightsquigarrow_{\mathcal{E}}\stackrel{\ell}{\longrightarrow}\mathcal{B}$ whenever $A\stackrel{\ell}{\longrightarrow}B\in\mathcal{B}$.

Let's see some examples.
For the process $\Omega_{\frac{1}{2}a}$ in Example~\ref{ex-3} one has $\Omega_{\frac{1}{2}a}\rightsquigarrow_{\mathcal{E}_2}\stackrel{a}{\longrightarrow}{\bf 0}$, where the regular $\epsilon$-tree is described by the right diagram in Example~\ref{ex-3}.
For the process $G$ in Example~\ref{ex-4} one has $G\rightsquigarrow_{\mathcal{E}_3}\stackrel{a}{\longrightarrow}{\bf 0}$, where the regular $\epsilon$-tree is described by the left diagram in Example~\ref{ex-4}.
For the process $H$ in Example~\ref{ex-5}, $H\rightsquigarrow_{\mathcal{E}_4}\stackrel{a}{\longrightarrow}{\bf 0}$ via the regular $\epsilon$-tree described by the left diagram, and $H\rightsquigarrow_{\mathcal{E}_4}\stackrel{b}{\longrightarrow}{\bf 0}$ via the regular $\epsilon$-tree described by the right diagram.

Now consider the situation where $A$ evolves into processes in $\mathcal{B}\in\left(\mathcal{P}/\mathcal{E}\right)\setminus\{[A]_{\mathcal{E}}\}$ with probability greater than $0$.
Suppose $L\stackrel{\coprod_{i\in [k]}p_i\tau}{\longrightarrow}\coprod_{i\in[k]}L_i$ such that $\exists i\in I.L_i\in\mathcal{B}$.
Define
\[\textsf{P}\left(L\stackrel{\coprod_{i\in [k]}p_i\tau}{\longrightarrow}\mathcal{B}\right)=\sum\left\{p_i\mid L\stackrel{p_i\tau}{\longrightarrow}L_i\in\mathcal{B} \wedge i\in I\right\}.\]
Define the {\em weighted} probability
\begin{equation}\label{2019-01-26}
\textsf{P}_{\mathcal{E}}\left(L\stackrel{\coprod_{i\in [k]}p_i\tau}{\longrightarrow}\mathcal{B}\right)=\textsf{P}\left(L\stackrel{\coprod_{i\in [k]}p_i\tau}{\longrightarrow}\mathcal{B}\right)\left/\left(1-\textsf{P}\left(L\stackrel{\coprod_{i\in [k]}p_i\tau}{\longrightarrow}[A]_{\mathcal{E}}\right)\right)\right..
\end{equation}
Intuitively~(\ref{2019-01-26}) is the probability that $L$ may leave the class $[A]_{\mathcal{E}}$ silently for elements of $\mathcal{B}$.
If one leaf of the regular $t_{\mathcal{E}}^A$ can do a silent transition that leaves $t_{\mathcal{E}}^A$ with a non-zero probability, we require that every leaf of $t_{\mathcal{E}}^A$ is capable of doing a silent transition that leaves $t_{\mathcal{E}}^A$ with that probability.
This probabilistic bisimulation property is observed in~\cite{BaierHermanns1997} in the simpler setting of the finite state fully probabilistic processes.
In our general setting a process may do several random silent transitions caused by different random combinators.
Suppose $\mathcal{B}\ne[A]$.
A {\em $q$-transition from $A$ to $\mathcal{B}$ with regard to $\mathcal{E}$} consists of a regular $\epsilon$-tree $t_{\mathcal{E}}^A$ of $A$ with regard to $\mathcal{E}$ and, for every leaf $L$ of $t_{\mathcal{E}}^A$, a collective silent transition $L\stackrel{\coprod_{i\in [k]}p_i\tau}{\longrightarrow}\coprod_{i\in[k]}L_i$ such that
\[\textsf{P}_{\mathcal{E}}\left(L\stackrel{\coprod_{i\in [k]}p_i\tau}{\longrightarrow}\mathcal{B}\right)=q.\]
We will write $A\rightsquigarrow_{\mathcal{E}}\stackrel{q}{\longrightarrow}\mathcal{B}$ if there is a $q$-transition from $A$ to $\mathcal{B}$ with regard to $\mathcal{E}$.
\begin{definition}
An equivalence $\mathcal{E}$ on $\mathcal{P}_{\mathrm{RCCS}}$ is a {\em branching bisimulation} if (\ref{ell-bi},\ref{ep-bi}) are valid.
\begin{enumerate}
\item \label{ell-bi}
If $B\,\mathcal{E}\,A\rightsquigarrow_{\mathcal{E}}\stackrel{\ell}{\longrightarrow}\mathcal{C}\in\mathcal{P}/\mathcal{E}$ such that $\ell\ne\tau\vee\mathcal{C}\not=[A]_{\mathcal{E}}$, then $B\rightsquigarrow_{\mathcal{E}}\stackrel{\ell}{\longrightarrow}\mathcal{C}$.
\item \label{ep-bi}
If $B\,\mathcal{E}\,A\rightsquigarrow_{\mathcal{E}}\stackrel{q}{\longrightarrow}\mathcal{C}\in\mathcal{P}/\mathcal{E}$ such that $\mathcal{C}\not=[A]_{\mathcal{E}}$, then $B\rightsquigarrow_{\mathcal{E}}\stackrel{q}{\longrightarrow}\mathcal{C}$.
\end{enumerate}
\end{definition}

Consider $\mu X.\left(a_1+\tau.(a_2+\tau.(\ldots(a_k+\tau.X)\ldots))\right)$.
The behaviour of the process can be pictured as a ring (the left diagram below), in which all nodes are equal~\cite{vanGlabbeekWeijland1996,FuYuxi2016}.
Consider a different process $\mu X.\left(\frac{1}{2}\tau.a_1\oplus\frac{1}{2}\tau.(\frac{1}{2}\tau.a_2\oplus\frac{1}{2}\tau.(\ldots(\frac{1}{2}\tau.a_k\oplus\frac{1}{2}\tau.X)\ldots))\right)$.
Its behaviour is pictured by the right diagram below.
No two nodes in the right ring can be in any branching bisimulation.
For example the top node in the ring can reach to the process $a_1$ with probability $1/2$, whereas the bottom node in the ring cannot reach to $a_1$ with probability $1/2$.
\begin{center}
\includegraphics[scale=0.75]{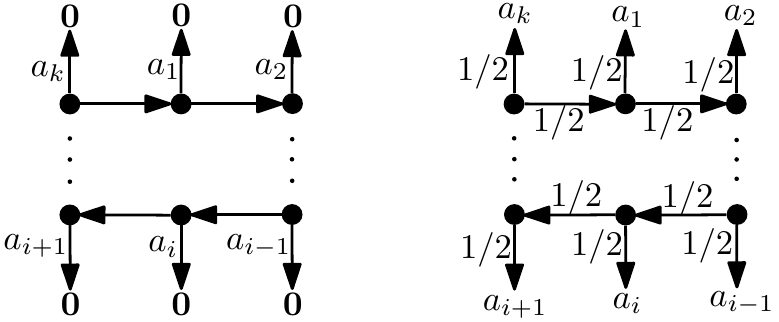}
\end{center}
The process $\Omega_{a}$ of Example~\ref{ex-1} and the process $\Omega_{\frac{1}{2}a}$ of Example~\ref{ex-3} cannot be in any codivergent branching bisimulation because the former is divergent whereas the latter is not.
For a relation $\mathcal{R}$ on $\mathcal{P}_{\mathrm{RCCS}}$, let $\mathcal{R}^{*}$ be the {\em equivalence closure} of $\mathcal{R}$.
Clearly $\left\{(\Omega_{\frac{1}{2}a},a)\right\}^{*}$ is a codivergent bisimulation.
And $\left\{(G,G_a)\right\}^{*}$ is a codivergent bisimulation, where $G$ is defined in Example~\ref{ex-4} and $G_a\stackrel{\rm def}{=}a+\tau.G$.
Also $\left\{(H,H_a),(H,H_b),(H,E)\right\}^{*}$ is a codivergent bisimulation, where $H$ is defined in Example~\ref{ex-5}, $H_a\stackrel{\rm def}{=}a+\tau.H$, $H_b\stackrel{\rm def}{=}b+\tau.H$ and $E\stackrel{\rm def}{=}\mu X.(a+b+\tau.X)$.

\section{Equality for Random Process}\label{Equality4RCCS}

The following lemma follows immediately from definition.
\begin{lemma}\label{2019-02-01}
If $\mathcal{E}_i$ is a codivergent equivalence for every $i\in I$, then so is $\left(\bigcup_{i\in I}\mathcal{E}_i\right)^{*}$.
\end{lemma}
The proof of the next fact is slightly complicated but standard.
\begin{proposition}\label{2018-12-18}
If $\mathcal{E}_i$ is a branching bisimulation for every $i\in I$, then so is $\left(\bigcup_{i\in I}\mathcal{E}_i\right)^{*}$.
\end{proposition}
\begin{proof}
Let $\mathcal{E}=\left(\bigcup_{i\in I}\mathcal{E}_i\right)^{*}$.
Assume that $(A_0,A_k)\in\mathcal{E}$ is due to $A_0\mathcal{E}_{i_1}A_1\mathcal{E}_{i_2}\ldots\mathcal{E}_{i_k}A_k$ for some $A_1,\ldots,A_{k-1}$.
Let $\mathcal{C}\in\mathcal{P}_{\mathrm{RCCS}}/\mathcal{E}$ such that $(\ell,\mathcal{C})\ne(\tau,[A_0]_{\mathcal{E}})$.
By definition there must be a family of pairwise disjoint equivalence classes $\left\{\mathcal{C}^{i_1}_j\right\}_{j\in J}$ of $\mathcal{E}_{i_1}$ such that
$\mathcal{C} = \bigcup_{j\in J}\mathcal{C}^{i_1}_{j}$.
Consider an $\ell$-transition $A_0\rightsquigarrow_{\mathcal{E}}\stackrel{\ell}{\longrightarrow}\mathcal{C}$.
It consists of an $\epsilon$-tree $t_{A_0}$ of $A_0$ with regard to $\mathcal{E}$ and, for every leaf $L$ of $t_{A_0}$, a transition $L\stackrel{\ell}{\longrightarrow}L'\in\mathcal{C}$.
We construct by induction on the structure of $t_{A_0}$ an $\ell$-transition $A_1\rightsquigarrow_{\mathcal{E}}\stackrel{\ell}{\longrightarrow}\mathcal{C}$.
The basic idea is to construct an $\epsilon$-tree, whose nodes are all in $[A_1]_{\mathcal{E}}$, for every edge of $t_{A_0}$.
By sticking these $\epsilon$-trees together we get an $\epsilon$-tree $t_{A_1}$ of $A_1$ with regard to $\mathcal{E}$.
Formally the bisimulation $A_1\rightsquigarrow_{\mathcal{E}}\stackrel{\ell}{\longrightarrow}\mathcal{C}$ can be derived by induction.
\begin{itemize}
\item
The root of $t_{A_0}$ has only one child $A_0'$.
By definition the edge from $A_0$ to $A_0'$ is labeled by $1$.
If $A_0'\in[A_0]_{\mathcal{E}_{i_1}}$, we construct $t_{A_1}$ by structural induction on the $\epsilon$-tree of $A_0'$.
If $A_0'\notin[A_0]_{\mathcal{E}_{i_1}}$ then $A_0\stackrel{\tau}{\longrightarrow}A_0'$ is bisimulated by some $\tau$-transition $A_1\rightsquigarrow_{\mathcal{E}_{i_1}}\stackrel{\tau}{\longrightarrow}[A_0']_{\mathcal{E}_{i_1}}$ consisting of an $\epsilon$-tree $t_{A_1}'$ of $A_1$ with regard to $\mathcal{E}_{i_1}$ and, for every leaf $B$ of $t_{A_1}'$, a transition $B\stackrel{\tau}{\longrightarrow}B'\mathcal{E}_{i_1}A_0'$ for some $B'$.
We then continue to construct an $\epsilon$-tree for each $B'$ by induction on the structure of the $\epsilon$-tree of $A_0'$.
\begin{figure*}[t]
\begin{center}
\includegraphics[scale=0.75]{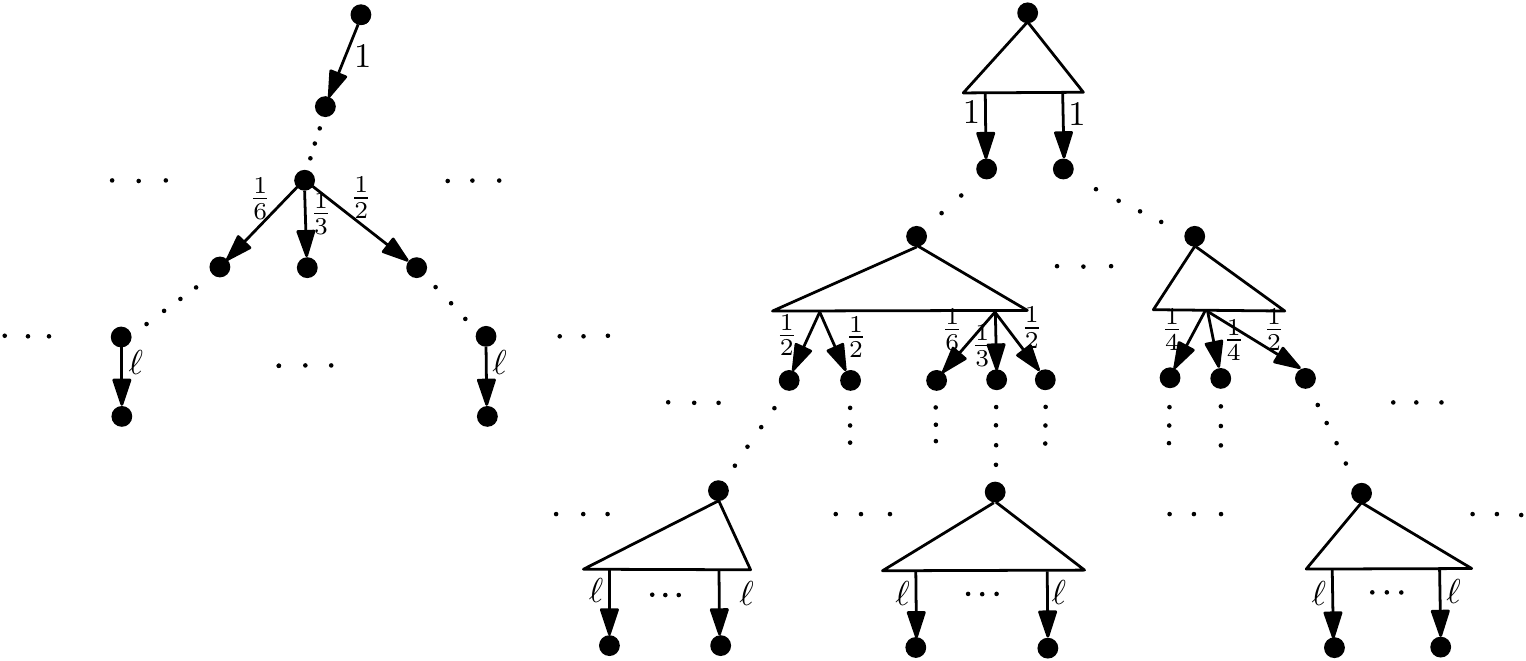}
\end{center}
\caption{Stepwise Bisimulation.}
\label{bi-tree}
\end{figure*}
\item
The root of $t_{A_0}$ has $h$ children $A_0^1,\ldots,A_0^h$ with the corresponding edges labeled by $p_1,\ldots,p_h$ respectively.
By definition
\[A_0\stackrel{\coprod_{j\in [h]}p_j\tau}{\longrightarrow}\coprod_{j\in [h]}A_0^j.\]
There are two cases.
In the first case $A_0^j\mathcal{E}_{i_1}A_0$ for all $j\in[k]$.
We construct $t_{A_1}'$ by structural induction on the $\epsilon$-tree of say $A_0^1$.
In the second case suppose without loss of generality that $A_0^1\notin[A_0]_{\mathcal{E}_{i_1}}$.
Let $q=\textsf{P}_{\mathcal{E}_{i_1}}\left(A_0\stackrel{\coprod_{i\in [h]}p_i\tau}{\longrightarrow}[A_0^1]_{\mathcal{E}_{i_1}}\right)$.
Then $A_1\rightsquigarrow_{\mathcal{E}_{i_1}}\stackrel{q}{\longrightarrow}[A_0']_{\mathcal{E}_{i_1}}$ by definition.
The $q$-transition consists of a regular $\epsilon$-tree $t_{A_1}'$ of $A_1$ with regard to $\mathcal{E}_{i_1}$ and, for each leaf $N$ of $t_{A_1}'$, a collective silent transition $N\stackrel{\coprod_{i'\in [h']}p_{i'}\tau}{\longrightarrow}\coprod_{i'\in[h']}N_{i'}$ such that
\[\textsf{P}_{\mathcal{E}_{i_1}}\left(N\stackrel{\coprod_{i'\in [h']}p_i\tau}{\longrightarrow}[A_0']_{\mathcal{E}_{i_1}}\right)=q.\]
For every process $N_{i''}$ the $q$-transition $A_1\rightsquigarrow_{\mathcal{E}_{i_1}}\stackrel{q}{\longrightarrow}[A_0']_{\mathcal{E}_{i_1}}$ reaches, we continue to construct an $\epsilon$-tree of $N_{i''}$ by induction on the structure of $A_0^1$.
\item
The root of $t_{A_0}$ does the transition $A_0\stackrel{\ell}{\longrightarrow}L'$.
Then $A_1\rightsquigarrow_{\mathcal{E}_{i_1}}\stackrel{\ell}{\longrightarrow}[L']_{\mathcal{E}_{i_1}}$ by definition.
\end{itemize}
In Figure~\ref{bi-tree} the left is a diagram for $A_0\rightsquigarrow_{\mathcal{E}}\stackrel{\ell}{\longrightarrow}\mathcal{C}$, while the right is a diagram for the stepwise bisimulation $A_1\rightsquigarrow_{\mathcal{E}}\stackrel{\ell}{\longrightarrow}\mathcal{C}$.
The above itemized cases are described by the upper, middle, and bottom parts of the diagrams respectively.
We still need to verify the regularity property.
Given $\varepsilon\in(0,1)$, there is a number $K_{\varepsilon}$ such that $1-\textsf{P}^{K_{\varepsilon}}(t_{A_0})<\varepsilon/2$.
Now every edge in $t_{A_0}{\upharpoonright}K_{\varepsilon}$ is bisimulated either vacuously or by an $\epsilon$-tree $t$.
There is a number $N_{\varepsilon}$ such that for every such $\epsilon$-tree $t$ it holds that $1-\textsf{P}^{N_{\varepsilon}}(t)<\frac{\varepsilon}{2K_{\varepsilon}}$.
It is not difficult to see that $1-\textsf{P}^{K_{\varepsilon}N_{\varepsilon}}(t_{A_1})<\varepsilon/2+\varepsilon/2=\varepsilon$.
Therefore $t_{A_1}$ is regular.
So $A_0\rightsquigarrow_{\mathcal{E}}\stackrel{\ell}{\longrightarrow}\mathcal{C}$ is bisimulated by $A_1\rightsquigarrow_{\mathcal{E}}\stackrel{\ell}{\longrightarrow}\mathcal{C}$.
For the same reason $A_1\rightsquigarrow_{\mathcal{E}}\stackrel{\ell}{\longrightarrow}\mathcal{C}$ is bisimulated by some $A_2\rightsquigarrow_{\mathcal{E}}\stackrel{\ell}{\longrightarrow}\mathcal{C}$.
We are done by induction.

We should also consider transitions of the form $A_0\rightsquigarrow_{\mathcal{E}}\stackrel{p}{\longrightarrow}\mathcal{C}$ for some $p\in(0,1)$, which can be treated in the same fashion.
\end{proof}

Proposition~\ref{2018-12-18} is reassuring.
We may now define the {\em equality on RCCS processes}, denoted by $=_{\mathrm{RCCS}}$, as the largest codivergent branching bisimulation on $\mathcal{P}_{\mathrm{RCCS}}$.
We abbreviate $=_{\mathrm{RCCS}}$ to $=$ in the rest of the section.

\begin{theorem}
The equality $=_{\mathrm{RCCS}}$ is a congruence.
\end{theorem}
\begin{proof}
It is easy to see that $=$ is closed under both the nondeterministic choice operation and the random choice operation.
Consider $\mathcal{R} \stackrel{\rm def}{=} \{(A\,|\,C,B\,|\,D) \mid A=B\wedge C=D\}$.
We prove that $\mathcal{R}^{\circ}\stackrel{\rm def}{=}\left(\mathcal{R}\;\cup=\right)^{*}$ is a codivergent branching bisimulation.
Suppose $A\,|\,C\;\mathcal{R}\;B\,|\,D\;=\ldots\;\mathcal{R}\;=E\,|\,F$
and $A\,|\,C\rightsquigarrow_{\mathcal{R}^{\circ}}\stackrel{\ell}{\longrightarrow}\mathcal{C}$ for some equivalence class $\mathcal{C}\in\mathcal{P}/\mathcal{R}^{\circ}$ such that $\ell\ne\tau\vee \mathcal{C}\ne[A\,|\,C]_{\mathcal{R}^{\circ}}$.
Let $t_{A\,|\,C}$ denote the $\epsilon$-tree of $A\,|\,C$ in the $\ell$-transition.
Using the technique explained in the proof of Proposition~\ref{2018-12-18} it is routine to build up an $\ell$-transition $B\,|\,D\rightsquigarrow_{\mathcal{R}^{\circ}}\stackrel{\ell}{\longrightarrow}\mathcal{C}$ that bisimulates $A\,|\,C\rightsquigarrow_{\mathcal{R}^{\circ}}\stackrel{\ell}{\longrightarrow}\mathcal{C}$.
This is inductively described as follows.
\begin{itemize}
\item
An edge from $A\,|\,C$ to $A'\,|\,C$ labeled $1$ is caused by a transition $A\stackrel{\tau}{\longrightarrow}_{=}A'$.
In this case $A'\,|\,C\;\mathcal{R}\;B\,|\,D$.
If it is caused by $A\stackrel{\coprod_{i\in I}p_i\tau}{\longrightarrow}\coprod_{i\in I}A_i$ such that $A_i=A$ for all $i\in I$, then obviously $A_i\,|\,C\;\mathcal{R}\;B\,|\,D$ for each $i\in I$.
\item
An edge from $A\,|\,C$ to $A'\,|\,C$ labeled $1$ is caused by a transition $A\stackrel{\tau}{\longrightarrow}A'\not= A$.
Then $B\rightsquigarrow_{=}\stackrel{\tau}{\longrightarrow}[A']_{=}$.
It should be clear that $B\,|\,D\rightsquigarrow_{\mathcal{R}}\stackrel{\tau}{\longrightarrow}[A'\,|\,C]_{\mathcal{R}}$.
\item
Suppose $A\stackrel{\coprod_{i\in[k]}p_i\tau}{\longrightarrow}\coprod_{i\in[k]}A_i$ and $A_1\ne A\ne A_2\ne A_1$ and $A_1\,|\,C=A_2\,|\,C\ne A\,|\,C$.
Define $q_1\stackrel{\rm def}{=}\textsf{P}_{=}\left(A\stackrel{\coprod_{i\in [k]}p_i\tau}{\longrightarrow}[A_1]_{=}\right)$ and $q_2\stackrel{\rm def}{=}\textsf{P}_{=}\left(A\stackrel{\coprod_{i\in [k]}p_i\tau}{\longrightarrow}[A_2]_{=}\right)$ and $q\stackrel{\rm def}{=}q_1+q_2$.
Then $q$ is equal to $\textsf{P}_{=}\left(A\,|\,C\stackrel{\coprod_{i\in [k]}p_i\tau}{\longrightarrow}[A_1\,|\,C]_{=}\right)$.
By assumption $B\rightsquigarrow_{=}\stackrel{q_1}{\longrightarrow}[A_1]_{=}$ and $B\rightsquigarrow_{=}\stackrel{q_2}{\longrightarrow}[A_2]_{=}$.
It follows that $B\,|\,D\rightsquigarrow_{\mathcal{R}}\stackrel{q}{\longrightarrow}[A_1\,|\,D]_{\mathcal{R}}$.
\item
An edge from $A\,|\,C$ to $A'\,|\,C'$ labeled $1$ is caused by $A\stackrel{a}{\longrightarrow}A'$ and $C\stackrel{\overline{a}}{\longrightarrow}C'$.
Then $B\rightsquigarrow_{=}\stackrel{a}{\longrightarrow}[A']_{=}$ and $D\rightsquigarrow_{=}\stackrel{\overline{a}}{\longrightarrow}[C']_{= }$.
It follows that $B\,|\,D\rightsquigarrow_{\mathcal{R}^{\circ}}\stackrel{a}{\longrightarrow}[A'\,|\,C']_{\mathcal{R}^{\circ}}$.
\end{itemize}
Thus $\mathcal{R}^{\circ}$ is a branching bisimulation.
The proof that $\mathcal{R}^{\circ}$ is codivergent is similar.

Next we argue that $=$ is closed under localization.
Define $\mathcal{S} \,\stackrel{\rm def}{=} \{((a)A,(a)B) \mid A=B\}$.
We show that $\mathcal{S}^{\circ} \,\stackrel{\rm def}{=} (\mathcal{S}\;\cup=)^{*}$ is a codivergent bisimulation.
Suppose $(a)A\,\mathcal{S}\,(a)B=\ldots\mathcal{S}=(a)D$ and that $t_{(a)A}$ is an $\epsilon$-tree of $(a)A$.
This $\epsilon$-tree is derived from a silent tree of $A$.
In the silent tree of $A$ an edge say $A'\stackrel{\tau}{\longrightarrow}A''$ may not be state-preserving, even though $(a)A'\stackrel{\tau}{\longrightarrow}(a)A''$ is state-preserving.
Suppose $B'=A'$ and $A'\stackrel{\tau}{\longrightarrow}A''$ is bisimulated by $B'\rightsquigarrow_{=}\stackrel{\tau}{\longrightarrow}[A'']_{=}$.
It is easily seen that $(a)B'\rightsquigarrow_{\mathcal{S}^{\circ}}\stackrel{\tau}{\longrightarrow}[(a)A'']_{\mathcal{S}^{\circ}}$ bisimulates $(a)A'\stackrel{\tau}{\longrightarrow}(a)A''$.
Arguing in this manner and using induction we show that if $\ell\ne\tau\vee[(a)A]_{S^{\circ}}$, then $(a)A\rightsquigarrow_{\mathcal{S}^{\circ}}\stackrel{\ell}{\longrightarrow}\mathcal{C}$ is bisimulated by some $(a)B\rightsquigarrow_{\mathcal{S}^{\circ}}\stackrel{\ell}{\longrightarrow}\mathcal{C}$.
The codivergence is easy.
\end{proof}

Referring to Example~\ref{ex-5} we see that $\mu X.(\frac{1}{2}\tau.(a+\tau.X)\oplus\frac{1}{2}\tau.(b+\tau.X))=\mu X.(a+b+\tau.X)$.

\section{Comment}\label{Comment}

We have proposed a model independent approach that turns a process model $\mathbb{M}$ into a randomized extension of $\mathbb{M}$.
We have demonstrated how to build up the bisimulation semantics of the randomized $\mathbb{M}$ on the bisimulation semantics of $\mathbb{M}$.
In our approach the bisimulation equality of the randomized $\mathbb{M}$ is a conservative extension of that of $\mathbb{M}$.
This is because $\epsilon$-trees of $A$ with regard to an equivalence $\mathcal{E}$ are the same as $A\Longrightarrow_{\mathcal{E}}$ if $A$ is a process in $\mathbb{M}$.
For example $A=_{\mathrm{CCS}}B$ if and only if $A=_{\mathrm{RCCS}}B$ for all $A,B\in\mathcal{P}_{\mathrm{CCS}}$.

The philosophy of the model independent method is that randomization is a computational property.
An external action cannot really be random because it depends on an open-ended environment.
An external action may appear random as a consequence of computational randomness.
Random computation is the reason; random interaction is a consequence.

The model independent approach can be investigated from the perspective of axiomatization~\cite{HanssonJonsson1989,JouSmolka1990,BaetenBergstraSmolka1995,StarkSmolka1999,BandiniSegala2001,DengPalamidessi2007},
equivalence checking algorithm~\cite{BaierHermanns1997,Philippou1LeeSokolsky2000}, logical characterization~\cite{SegalaLynch1994}, other equivalences say testing equivalence~\cite{DeNicolaHennessy1984,Hennessy1988}.
In the light of previous works on these issues, we expect that all the investigations are routine exercises.

%________________________________________________________
\vspace*{3.5mm}

\noindent{\bf Acknowledgment}.
We are grateful to the support from NSFC (61472239, 91318301).
We would like to thank Yuxin Deng and the members of BASICS for discussions on the issue.

%________________________________________________________


\begin{thebibliography}{10}
\bibitem{AndovaWillemse2006}
S.~Andova and A.~Willemse.
\newblock Branching Bisimulation for Probabilistic Systems: Characteristics and Decidability.
\newblock {\em Theoretical Computer Science}, 356:325--355, 2006.

\bibitem{AroraLundMotwaniSudanSzegedy1992}
S.~Arora, C.~Lund, R.~Motwani, M.~Sudan, and M.~Szegedy.
\newblock Proof Verification and the Hardness of Approximation Problems.
\newblock {\em J. ACM}, 1998 ({\em FOCS'92}).

\bibitem{Babai1985}
L.~Babai.
\newblock Trading Group Theory for Randomness.
\newblock In {\em STOC'85}, ACM, 1985.

\bibitem{BabaiFortnowLund1991}
L.~Babai, L.~Fortnow, and L.~Lund.
\newblock Nondeterministic Exponential Time Has Two Prover Interactive Protocols.
\newblock {\em Computational Complexity}, 1991 ({\em FOCS’90}).

\bibitem{Baeten1996}
J.~Baeten.
\newblock Branching bisimilarity is an equivalence indeed.
\newblock {\em Information Processing Letters} 58:141--147, 1996.

\bibitem{BaetenBergstraSmolka1995}
J.~Baeten, J.~Bergstra, and S.~Smolka.
\newblock Axiomatizing Probabilistic Processes: ACP with generative probabilities.
\newblock {\em Information and Computation}, 122:234--255, 1995.

\bibitem{BandiniSegala2001}
E.~Bandini and R.~Segala.
\newblock Axiomatization for Probabilistic Bisimulation.
\newblock In {\em ICALP'2001}, Lecture Notes in Computer Science 2076, pages 370--381, Springer, 2001.

\bibitem{BaierHermanns1997}
C.~Bauer and H.~Hermanns.
\newblock Weak bisimulation for fully probabilistic processes.
\newblock In {\em CAV'97}, Lecture Notes in Computer Science 1254, pages 119--130, Springer, 1997.

\bibitem{Ben-OrGoldwasserKilianWigderson1988}
M.~Ben-Or, S.~Goldwasser, J.~Kilian, and A.~Wigderson.
\newblock Multi-Prover Interactive Proofs: How to Remove Intractability Assumptions.
\newblock In {\em STOC'88}, ACM, 1988.

\bibitem{Deng2015}
Y.~Deng.
\newblock {\em Semantics of Probabilistic Processes: An Operational Approach}.
\newblock Springer-Verlag and Shanghai Jiao Tong University Press, 2015.

\bibitem{DengPalamidessi2007}
Y.~Deng and C.~Palamidessi.
\newblock Axiomatizations for Probabilistic Finite-State Behaviours.
\newblock {\em Theoretical Computer Science}, 2007.

\bibitem{DeNicolaHennessy1984}
R.~De~Nicola, M.~Hennessy.
\newblock Testing equivalence for processes.
\newblock {\em Theoretical Computer Science}, 34:83--133, 1984.

\bibitem{FortnowRompelSipser1994}
L.~Fortnow, J.~Rompel, and M.~Sipser.
\newblock On the Power of Multi-Prover Interactive Protocols.
\newblock {\em Theoretical Computer Science}, 21:545--557, 1994.

\bibitem{FuYuxi2016}
Y.~Fu.
\newblock Theory of Interaction.
\newblock {\em Theoretical Computer Science}, 611:1--49, 2016.

\bibitem{FuYuxi2017}
Y.~Fu.
\newblock The Universal Process.
\newblock {\em Logical Methods in Computer Science}, 13:1-23, 2017.

\bibitem{GoldwasserMicaliRackoff1985}
S.~Goldwasser, S.~Micali, and C.~Rackoff.
\newblock The Knowledge Complexity of Interactive Proofs.
\newblock In {\em STOC'85}, ACM, 1985.

\bibitem{HanssonJonsson1989}
H.~Hansson and B.~Jonsson.
\newblock A framework for reasoning for reasoning about time and reliability.
\newblock In {\em IEEE Symposium on Real-Time Systems}, IEEE, 1989.

\bibitem{Hennessy1988}
M.~Hennessy.
\newblock {\em An Algebraic Theory of Processes}.
\newblock MIT Press, Cambridge, MA, 1988.

\bibitem{JouSmolka1990}
C.~Jou and S.~Smolka.
\newblock Equivalences and complete axiomatizations for probabilistic processes.
\newblock In {\em CONCUR'90}, Lecture Notes in Computer Science 458, pages 367--383, 1990.

\bibitem{LarsenSkou1989}
K.~Larsen and A.~Skou.
\newblock Bisimulation Through Probabilistic Testing.
\newblock In {\em POPL'89}, 344--352, ACM, 1989.

\bibitem{Milner1989}
R.~Milner.
\newblock {\em Communication and Concurrency}.
\newblock Prentice Hall, 1989.

\bibitem{Park1981}
D.~Park.
\newblock Concurrency and Automata on Infinite Sequences.
\newblock In {\em TCS'81}, Lecture Notes in Computer Science 104, pages 167--183, Springer, 1981.

\bibitem{Philippou1LeeSokolsky2000}
A.~Philippou1, I.~Lee, and O.~Sokolsky.
\newblock Weak Bisimulation for Probabilistic Systems.
\newblock In {\em CONCUR'00}, Lecture Notes in Computer Science 1877, pages 334--349, Springer, 2000.

\bibitem{Priese1978}
L.~Priese.
\newblock On the concept of simulation in asynchronous, concurrent systems.
\newblock {\em Prog. Cybern. Syst. Res.}, 7:85–92, 1978.

\bibitem{Segala1995}
R.~Segala.
\newblock {\em Modelling and Verification of Randomized Distributed Rela-Time Systems}.
\newblock PhD Thesis, MIT, Dept. of EECS, 1995.

\bibitem{Segala2006}
R.~Segala.
\newblock Probability and Nondeterminism in Operational Models of Concurrency.
\newblock In {\em CONCUR'06}, Lecture Notes in Computer Science 4137, pages. 64--78, Springer, 2006.

\bibitem{SegalaLynch1994}
R.~Segala and N.~Lynch.
\newblock Probabilistic Simulations for Probabilistic Processes.
\newblock In {\em CONCUR'94}, Lecture Notes in Computer Science 836, pages 481--496, Springer, 1994.

\bibitem{StarkSmolka1999}
E.~Stark and S.~Smolka.
\newblock A complete axiom system for finite-state probabilistic processes.
\newblock In {\em Language and Interaction: Essays in Honour of Robin Milner}, 1999.

\bibitem{Vadhan2012}
S.~Vadhan.
\newblock {\em Pseudorandomness}.
\newblock Foundations and Trends in Theoretical Computer Science. 7. Now Publishers Inc., 2012.

\bibitem{vanGlabbeekSmolkaSteffen1995}
R.~van Glabbeek S.~Smolka and B.~Steffen.
\newblock Reactive, generative, and stratified models of probabilistic processes.
\newblock {\em Information and Computation}, 3:59--80, 1995.

\bibitem{vanGlabbeekWeijland1989-first-paper-bb}
R.~van Glabbeek and W.~Weijland.
\newblock Branching Time and Abstraction in Bisimulation Semantics.
\newblock In {\em Information Processing'89}, pages 613--618, North-Holland, 1989.

\bibitem{vanGlabbeekWeijland1996}
R.~van Glabbeek and W.~Weijland.
\newblock Branching Time and Abstraction in Bisimulation Semantics.
\newblock {\em Journal of ACM}, 3:555--600, 1996.

\bibitem{Vardi1985}
M.~Vardi.
\newblock Automatic verification of probabilistic concurrent finite-state programs.
\newblock In {\em FOCS'85}, pages 327--338, IEEE, 1985.

\bibitem{WangLarsen1992}
Y.~Wang and K.~Larsen.
\newblock Testing probabilistic and nondeterministic processes.
\newblock In {\em Protocol Specification, Testing and Verification XII}, pages 47--61, 1992.
\end{thebibliography}
\end{document}